\documentclass[sn-chicago]{sn-jnl}

\jyear{2022}%

\theoremstyle{thmstyleone}%
\newtheorem{theorem}{Theorem}
\newtheorem{proposition}[theorem]{Proposition}%

\theoremstyle{thmstyletwo}%

\theoremstyle{thmstylethree}%

\usepackage{adjustbox}
\usepackage{dsfont}

\begin{document}

\title[Similarity matrix average for aggregating multiplex networks]{Similarity matrix average for aggregating multiplex networks}


\author*[1,2]{\fnm{Federica} \sur{Baccini}}\email{federica.baccini@phd.unipi.it}

\author[3]{\fnm{Lucio} \sur{Barabesi}}\email{lucio.barabesi@unisi.it}

\author[3]{\fnm{Eugenio} \sur{Petrovich}}\email{eugenio.petrovich@unisi.it}


\affil*[1]{\orgdiv{Department of Computer Science}, \orgname{University of Pisa}, \state{} \orgaddress{\street{Largo Bruno Pontecorvo, 3}, \city{Pisa}, \postcode{56127}, \country{Italy}}}

\affil[2]{\orgdiv{Institute for Informatics and Telematics}, \orgname{CNR},\state{} \orgaddress{\street{Via Giuseppe Moruzzi, 1}, \city{Pisa}, \postcode{56124}, \country{Italy}}}

\affil[3]{\orgdiv{Department of Economics and Statistics}, \orgname{University of Siena}, \orgaddress{\street{Piazza San Francesco, 7}, \city{Siena}, \postcode{53100},\state{} \country{Italy}}}

\abstract{We introduce a methodology based on averaging similarity matrices with the aim of integrating the layers of a multiplex network into a single monoplex network. Multiplex networks are adopted for modelling a wide variety of real-world frameworks, such as multi-type relations in social, economic and biological structures. More specifically, multiplex networks are used when relations of different nature (layers) arise between a set of elements from a given population (nodes).  A possible approach for investigating multiplex networks consists in aggregating the different layers in a single network (monoplex) which is a valid representation -- in some sense -- of all the layers. In order to obtain such an aggregated network, we propose a theoretical approach -- along with its practical implementation -- which stems on the concept of similarity matrix average. This methodology is finally applied to a multiplex similarity network of statistical journals, where the three considered layers express the similarity of the journals based on co-citations, common authors and common editors, respectively.}


\maketitle

\section{Introduction} \label{sec::intro}
Multilayer networks constitute an increasing active research
topic with applications in many different disciplines, such as social or biological sciences. A multilayer network is a collection of individual
networks -- referred to as layers -- each containing its own nodes and edges,
with additional edges between the various layers. For more details on
multilayer networks, see the recent monographs by \cite{bianconi2018multilayer}, \cite{dedomenico2022multilayer} and \cite{dickison2016multilayer}, or the survey paper by \cite{kivela2014multilayer}.

A noteworthy special case of the multilayer network is the so-called multiplex
network, which is characterized by the same set of nodes in each layer (see e.g. \cite{newman2018networks}, Section 6.7). Indeed, in
multiplex networks a unique node type is present, even if different edge
types may occur. In such a case, the linking edges between the various layers trivially connect replicas of the same node in the different layers, even
if such edges (interlayers) are omitted in practice for simplicity.
Characteristic examples of multiplex networks are social networks, where
the nodes are the individuals in a well-defined community with different
types of relational ties between them (such as friendship, family or
co-worker connections). Each type of relation is represented as a separate layer -- see \cite{bianconi2018multilayer} for real examples of social, on-line social, economic and financial
multiplex networks. A multiplex network illustration in the field of scientometrics
is provided by \cite{baccini2022similarity}. Actually, these authors consider multiplex networks of journals in different disciplines constituted by three layers based on co-citations, common
authors and common editors (see also \cite{baccini2020intellectual}). A further
instance of multiplex network modelling can be found in the field of epigenetics, as shown by \cite{baccini_graph-based_2022}. Indeed,  \cite{baccini_graph-based_2022} consider a multiplex network of blood cell phenotypes, where the layers are based on different epigenetic modifications.

A tool for investigating the structure of multiplex networks is based on the
``aggregation'' (in some appropriate way) of the different layers in order to
obtain a suitable monoplex network, i.e. a single layer. The resulting monoplex network is
useful to interpret the strong connections between nodes and -- eventually --
to implement a cluster analysis for searching cohesive groups in the original
multiplex network. In addition, the analysis of the correlation between the
monoplex and the single layers may be helpful to detect the global community
structure. In such a setting, the target is focused on how
to aggregate the different layers. A straightforward procedure for
implementing an aggregated network may be based on a weighted average of the
adjacency matrices corresponding to the layers (see e.g. \cite{kivela2014multilayer} and references therein). Even if this strategy actually produces a
single weighted adjacency matrix which characterizes the monoplex, it is
prone to some drawbacks, since some layers may be over-represented. Moreover, the
choice of the weights is subjective \citep{kivela2014multilayer}. A more
appropriate approach would be to consider similarity matrices -- with suitable features
-- obtained from the adjacency matrices corresponding to the
single layers, and then carry out an ``average'' of these similarity matrices
according to well-defined mathematical properties. Subsequently, by means of an appropriate similarity matrix average, the aggregated network is obtained.

In the present paper, we consider some recent advances in the theory of
barycenters of objects lying in abstract spaces finalized to the concept of
average for positive definite matrices \citep{alvarez2016fixed, bhatia2009positive}. More precisely, starting from two-mode multilayer
networks, we first show that the corresponding similarity matrices based on the
commonly-adopted similarity measures -- such as the Jaccard, the cosine, and
the SimRank similarity -- belong to the space of completely positive
matrices, which is a subspace of positive definite matrices. Subsequently, in order to
aggregate multiplex networks, we consider the Fr\a'echet mean criterion (see e.g.
\cite{bacak2014computing}) with three different metric choices -- i.e. the
classical Frobenius metric, the Riemannian metric, and the Wasserstein metric
-- to achieve appropriate definitions of average for positive definite
matrices. In addition, we remark the properties of the averages obtained with the three metrics. We
also discuss the choice of the weights to be adopted in the Fr\a'echet
criterion. Finally, we illustrate the theoretical findings with an
application to the aforementioned multiplex networks considered by \cite{baccini2022similarity} in the field of scientometrics.

The paper is organized as follows. Section \ref{sec::notation} contains some preliminaries and notation. In Section \ref{sec::simproperties} we show that similarity matrices based on Jaccard, cosine and
SimRank similarity are completely positive. In Section \ref{sec::barycenters} we present the theoretical framework to compute the averages of the similarity matrices connected to the multiplex networks. In Section \ref{sec::weightchoice} we propose some choices for the weights involved in the averages. Section \ref{sec::implementation} presents the details of the code implementation for the practical computation of the averages. Section \ref{sec::journals} presents the application to the multiplex network involving statistical journals. Finally, Section \ref{sec::conclusion} draws some conclusions.

\section{Some notations and preliminaries} \label{sec::notation}
Let us assume that $\mathcal{P}_n$ is
the space of symmetric positive semidefinite matrices of order $n$,
i.e.
\begin{equation*}\label{def::Pn}
\mathcal{P}_n=\{\boldsymbol{X}\in\mathbb{R}^{n\times
n}:\boldsymbol{X}=\boldsymbol{X}^{\text{T}},\boldsymbol{X}\succeq\boldsymbol{0}\}\text{ ,}
\end{equation*}
where $\boldsymbol{X}\succeq\boldsymbol{0}$ denotes that the eigenvalues of $\boldsymbol{X}$ are
nonnegative. The space $\mathcal{P}_n$ is a differentiable manifold in the
set of Hermitian matrices (for more details on this class of matrices, see
\cite{bhatia2009positive}, Chapter 6). The Frobenius inner product on $\mathcal{P}_n$
is defined by
$\langle\boldsymbol{X}$,$\boldsymbol{Y}\rangle_{\text{F}}=$tr$(\boldsymbol{X}^{\text{T}}\boldsymbol{Y})$ for
$\boldsymbol{X},\boldsymbol{Y}\in\mathcal{P}_n$, and the associated Frobenius norm is given by
$\|\boldsymbol{X}\|_{\text{F}}=$tr$(\boldsymbol{X}^{\text{T}}\boldsymbol{X})^{1/2}$. If $\boldsymbol{X}\in\mathcal{P}_n$,
$\boldsymbol{X}^{1/2}$ denotes the unique positive definite matrix such that
$\boldsymbol{X}=\boldsymbol{X}^{1/2}\boldsymbol{X}^{1/2}$, while the corresponding normalized matrix is given by
$\lambda_1(\boldsymbol{X})^{-1} \boldsymbol{X}$, where $\lambda_1(\boldsymbol{X})>0$ represents the largest
eigenvalue of $\boldsymbol{X}$. Finally, if $\boldsymbol{X}=(x_{ij})$ is a square matrix of order
$n$, we assume that diag$(\boldsymbol{X})=(\delta_{ij}x_{ij})$, while if
$\boldsymbol{x}=(\boldsymbol{x}_1,\ldots,\boldsymbol{x}_n)^{\text{T}}$ is a vector then
diag$(\boldsymbol{x})=(\delta_{ij}x_i)$, where $\delta_{ij}$ represents the Kronecker
delta.

Let us consider the subspace $\mathcal{C}\mathcal{P}_n\subset\mathcal{P}_n$
of completely positive matrices, i.e.
\begin{equation*}
\mathcal{C}\mathcal{P}_n=\{\boldsymbol{X}\in\mathbb{R}^{n\times
n}:\boldsymbol{X}=\boldsymbol{Y}^{\text{T}}\boldsymbol{Y},\exists\boldsymbol{Y}\in\mathbb{R}^{p\times
n},\boldsymbol{Y}\geq\boldsymbol{0}\}\text{ ,}
\end{equation*}
where $\boldsymbol{X}\geq\boldsymbol{0}$ denotes that the entries of $\boldsymbol{X}$ are nonnegative.
For the properties of this class of matrices, see the monographs by \cite{berman2003completely}, \cite{shaked2021copositive} and  \cite{johnson2020matrix}. A necessary condition for a matrix $\boldsymbol{X}$ to be completely positive
is that $\boldsymbol{X}\succeq\boldsymbol{0}$ and $\boldsymbol{X}\geq\boldsymbol{0}$, even if --
contrary to intuition -- this condition is not generally sufficient \citep{berman2003completely}. In addition, \cite{berman2003completely} show that $\mathcal{C}\mathcal{P}_n$ is a closed convex cone in the
class of Hermitian matrices -- and hence in $\mathcal{P}_n$. The following Proposition provides some results
on completely positive matrices which will be helpful in the following
sections.
\begin{proposition}\label{prop:1}
By assuming that $\boldsymbol{X},\boldsymbol{Y} \in\mathcal{C}\mathcal{P}_n$, it holds:

$(i)$ $\boldsymbol{X}+\boldsymbol{Y}$ $\in\mathcal{C}\mathcal{P}_n$; 

$(ii)$ $\boldsymbol{X}\odot\boldsymbol{Y} \in\mathcal{C}\mathcal{P}_n$, where the symbol $\odot$
denotes the Hadamard product;

$(iii)$ $\boldsymbol{Z}^{\text{T}}\boldsymbol{X}\boldsymbol{Z} \in\mathcal{C}\mathcal{P}_n$ if $\boldsymbol{Z}$ is a square matrix
of order $n$ such that $\boldsymbol{Z} \geq \boldsymbol{0}$;
\end{proposition}
\begin{proof}See Corollary 2.1., 2.2 and Proposition 2.2 by \cite{berman2003completely}.
\end{proof}

\section{Completely positive similarity matrices from two-mode networks}\label{sec::simproperties}

We consider two-mode (bipartite) networks, i.e. networks displaying two types of
nodes where edges solely tie nodes of different type (for
more details on bipartite networks, see e.g. \cite{newman2018networks}, Chapter 6). This
kind of networks is commonly adopted to describe the membership of a set of
$p$ items to $n$ groups. The items are represented by the first set of nodes,
while the groups are represented by the second set of nodes -- and the edges
connect the items to the groups they belong to. 
A bipartite network
captures exactly the same information as a hypergraph, even if for most
purposes bipartite networks are more convenient. Indeed, most of the social and
biological networks belong to this family.

A two-mode network can be characterized by an incidence matrix
$\boldsymbol{B}=(b_{ij})$ of order $(p\times n)$, where
\begin{equation*}\label{def::Bbipartite}
b_{ij}=
\begin{cases}
1&\text{if item }i\text{ belongs to group }j\\
0&\text{otherwise .}
\end{cases}
\end{equation*}
To avoid triviality, we assume that there exist $i\in\{1,\ldots,p\}$ such
that $b_{ij}=1$ for each $j=1,\ldots,n$, i.e. there is at least one item for
each group. Obviously, it holds that $\boldsymbol{B}\geq\boldsymbol{0}$. In order to
provide a practical illustration in the scientometric
setting, the network of journals and editors in a given discipline may be
considered as an example of a two-mode bipartite network with $p$ editors as items and $n$
journals as groups. In such a case, $b_{ij}$ constitutes the indicator
variable for the membership of the $i$-th editor to the $j$-th journal. 
On the basis of the incidence matrix $\boldsymbol{B}$, the one-mode projection with respect to
groups of the two-mode bipartite network can be considered. Hence, if
$\boldsymbol{G} =(g_{ij})$ is the matrix of order $(n\times n)$ such that
$\boldsymbol{G} =\boldsymbol{B}^{\text{T}}\boldsymbol{B}$, it holds that $g_{ij}=\sum_{k=1}^p b_{ki} b_{kj}$, i.e.
the entry $g_{ij}$ gives the total number of items which are in both the
$i$-th group and the $j$-th group for $i,j=1,\ldots,n$. Consequently, $g_{ii}$
gives the number of items in the $i$-th group. As an example, in the network
of journals and editors, $g_{ij}$ provides the number of editors belonging to
the board of both the $i$-th journal and the $j$-th journal, while $g_{ii}$
is the number of editors in the board of the $i$-th journal. This definition of $\boldsymbol{G}$ implies that $\boldsymbol{G}\in\mathcal{C}\mathcal{P}_n$. 

The weighted adjacency matrix $\boldsymbol{Z}=(z_{ij})$ of order $(n\times n)$
corresponding to the weighted one-mode projection is defined as
$\boldsymbol{Z}=\boldsymbol{G}-$diag$(\boldsymbol{G})$. Thus, $z_{ij}$ denotes the weight of the edge between
the $i$-th group node and the $j$-th group node. The corresponding unweighted
adjacency matrix $\boldsymbol{A}=(a_{ij})$ of order $(n\times n)$ is such that
$a_{ij}=\mathds{1}_{\{z_{ij}>0\}}$ for $i,j=1,\ldots,n$, where
$\mathds{1}_E$ is the indicator function of the set $E$. Hence,
$a_{ij}$ is binary-valued, as it assumes value $1$ if there is an edge between the
$i$-th group node and the $j$-th group node, $0$ otherwise. As an example, in the network of
journals and editors, $a_{ij}=1$ if the $i$-th journal and the $j$-th journal
share at least one editor; $a_{ij}=0$ otherwise. It is clear
that $\boldsymbol{Z}\geq\boldsymbol{0}$ and $\boldsymbol{A}\geq\boldsymbol{0}$. Finally, we consider
the matrix $\boldsymbol{F} = (f_{ij})$ given by $\boldsymbol{F} = \boldsymbol{A}^{\text{T}}\boldsymbol{A} = \boldsymbol{A}^2$. Thus, since
$f_{ij}=\sum_{k=1}^na_{ki}a_{kj}$, the entry $f_{ij}$ gives the total number
of common neighbours between the $i$-th group node and the $j$-th group node
for $i,j=1,\ldots,n$. Moreover, since $a_{ij}^2=a_{ij}$ for $i,j=1,\ldots,n$,
the entry $f_{ii}$ gives the degree corresponding to the $i$-th group node.
On the basis of its definition, it is also follows that
$\boldsymbol{F}\in\mathcal{C}\mathcal{P}_n$. 

The introduction of the matrices $\boldsymbol{G}\text{, } \boldsymbol{Z} \text{ and } \boldsymbol{F}$ permits to analyse some common choices of similarity measures that allow to build one-mode networks.
Our aim is to show
that the commonly adopted similarity matrices are completely positive. In
\cite{newman2018networks}, Section 7.6, several measures of structural and regular
equivalence -- i.e. similarity -- for a network are indicated. Concerning the normalized
similarity measures of structural equivalence, $g_{ij}$ belongs itself to this class for similarity measures. Therefore, $\boldsymbol{G}$ may be
considered as a similarity matrix, even if it may be not suitable to adopt
this option in practice. Some authors implement the similarity matrices of
structural equivalence on the basis of the matrix $\boldsymbol{G}$ (see e.g. \cite{leydesdorff2008normalization}), while others suggest the matrix $\boldsymbol{F}$ to the same aim (see e.g. \cite{newman2018networks}, Section 7.6). Since both these matrices are completely positive, we
provide the following results by adopting $\boldsymbol{G}$ -- even if they can be also
obtained by using $\boldsymbol{F}$ in place of $\boldsymbol{G}$.

In the context of similarity matrices of structural equivalence, \cite{newman2018networks} emphasizes that the well-known Jaccard coefficient and the cosine
similarity coefficient are the most widely used for the practical analysis of
networks (see also \cite{eck2009normalize}). We first consider the
similarity matrix $\boldsymbol{J} =(J_{ij})$ of order $(n\times n)$ which is based on the
Jaccard coefficient. More precisely, $J_{ij}$ constitutes the Jaccard
coefficient between the $i$-th group and the $j$-th group, i.e.
\begin{equation*}\label{eq:jacc}
J_{ij}=\frac{g_{ij}}{g_{ii}+g_{jj}-g_{ij}}\text{ .}
\end{equation*}
It should be remarked that the quantity $(g_{ii}+g_{jj}-g_{ij})$ in the
denominator of $J_{ij}$ actually represents the total number of items belonging to the $i$-th group or to the $j$-th group. Hence, $J_{ij}$ is a
normalized similarity measure in $[0,1]$ obtained by considering the common
items to the two groups divided by the total number of distinct items in the
two groups. Since $\boldsymbol{G}$ is symmetric, it can be easily proved that $\boldsymbol{J}$ is in turn
symmetric. The following Proposition provides the target result on $\boldsymbol{J}$.

\begin{proposition}\label{prop:2}
The matrix $\boldsymbol{J}$ is completely positive.
\end{proposition}
\begin{proof}
Let us consider the matrix $\boldsymbol{U}=(u_{ij})$ of order $(n\times n)$ where
\begin{equation*}
u_{ij}=p-g_{ii}-g_{jj}+g_{ij}\text{ ,}
\end{equation*}
i.e. $u_{ij}$ gives the total number of items which do not belong to the
$i$-th group and/or to the $j$-th group. Since $\boldsymbol{G}=\boldsymbol{B}^{\text{T}}\boldsymbol{B}$ and since
$b_{ij}^2=b_{ij}$ for $i,j=1,\ldots,n$, it reads
\begin{equation*}
u_{ij}=\sum_{k=1}^p(1-b_{ki}-b_{kj}+b_{ki}b_{kj})=%
\sum_{k=1}^p(1-b_{ki})(1-b_{kj})\text{ .}
\end{equation*}
Hence, it holds
$\boldsymbol{U}=(\boldsymbol{1}_p\boldsymbol{1}_n^{\text{T}}-\boldsymbol{B})^{\text{T}}(%
\boldsymbol{1}_p\boldsymbol{1}_n^{\text{T}}-\boldsymbol{B})$, where $\boldsymbol{1}_p$ denotes the column vector of length $p$ with all entries equal to $1$, and $\boldsymbol{1}_n^{\text{T}}$ is the row vector of length $n$ with elements all equal to $1$. since
$\boldsymbol{1}_p\boldsymbol{1}_n^{\text{T}}-\boldsymbol{B}\geq\boldsymbol{0}$, we have that $\boldsymbol{U}\in\mathcal{C}\mathcal{P}_n$.
Moreover, let us consider the matrix $\boldsymbol{V}=(v_{ij})$ of order $(n\times n)$ such
that
\begin{equation*}
v_{ij}=\frac{1}{1-u_{ij}/p}=\sum_{k=0}^\infty\frac{1}{p^k}\,u_{ij}^k\text{ ,}
\end{equation*}
where the last equality is obtained from the Geometric series, since it holds
$0<u_{ij}<p$ on the basis of the definition of $u_{ij}$. The matrix $\boldsymbol{V}$ may be
expressed as
\begin{equation*}
\boldsymbol{V}=\sum_{k=0}^\infty\frac{1}{p^k}\,\boldsymbol{U}_k\text{ ,}
\end{equation*}
where the matrix $\boldsymbol{U}_k$ is recursively defined as $\boldsymbol{U}_k = \boldsymbol{U}_{k-1}\odot\boldsymbol{U}$ with
the initial condition $\boldsymbol{U}_0 = \boldsymbol{1}_n\boldsymbol{1}_n^{\text{T}}$.
Since the Hadamard product of completely positive matrices is in turn a
completely positive matrix from $(ii)$ of Proposition 1 and
$\boldsymbol{U}_0\in\mathcal{C}\mathcal{P}_n$, it also holds that
$\boldsymbol{U}_k\in\mathcal{C}\mathcal{P}_n$. In addition, since also the sum of
completely positive matrices is a completely positive matrix from $(i)$ of
Proposition 1, and $\mathcal{C}\mathcal{P}_n$ is a closed cone, it follows
that $\boldsymbol{V}\in\mathcal{C}\mathcal{P}_n$. Thus, since
\begin{equation*}
\boldsymbol{J}=\frac{1}{p}\,\boldsymbol{G}\odot\boldsymbol{V}
\end{equation*}
and $\boldsymbol{G}\in\mathcal{C}\mathcal{P}_n$, we conclude that
$\boldsymbol{J}\in\mathcal{C}\mathcal{P}_n$.
\end{proof}

In the framework of structural equivalence, a further suitable similarity matrix is obtained from the cosine similarity coefficient. This choice yields a similarity matrix $\boldsymbol{C} = (c_{ij})$ of order $(n\times n)$, where the entry
$c_{ij}$ represents the cosine similarity coefficient
\begin{equation}\label{eq::cossim}
c_{ij}=\frac{g_{ij}}{(g_{ii}g_{jj})^{1/2}}\text{ .}
\end{equation}
The appellative of cosine coefficient derives from its morphologhy
as a ratio of an inner product of two vectors to the product of the
corresponding norms. Hence, on the basis of the Cauchy-Schwarz inequality,
$c_{ij}$ is a normalized similarity measure in $[0,1]$. Moreover, since $\boldsymbol{G}$ is
symmetric, it can be easily shown that $\boldsymbol{C}$ is in turn symmetric. The following
Proposition gives a result analogous to Proposition 3 for the matrix $\boldsymbol{C}$.

\begin{proposition}\label{prop:3}
The matrix $\boldsymbol{C}$ is completely positive.
\end{proposition}
\begin{proof}Let us consider the vector
$\boldsymbol{v} = (g_{11}^{-1/2},\ldots,g_{nn}^{-1/2})^{\text{T}}$ of order $n$ and the
diagonal matrix $\boldsymbol{U}=$diag$(\boldsymbol{v})$. The matrix $\boldsymbol{C}$ can be expressed as
\begin{equation*}
\boldsymbol{C}=\boldsymbol{U}^{\text{T}}\boldsymbol{GU}=\boldsymbol{U}^{\text{T}}\boldsymbol{B}^{\text{T}}%
\boldsymbol{BU}=(\boldsymbol{BU})^{\text{T}}\boldsymbol{BU},
\end{equation*}
and hence $\boldsymbol{C}\in\mathcal{C}\mathcal{P}_n$ since
$\boldsymbol{BU}\geq\boldsymbol{0}$.
\end{proof}

In the framework of normalized similarity measures of regular equivalence, the SimRank similarity measure,
originally introduced by \cite{jeh2002simrank}, is often considered (see
\cite{newman2018networks}, Section 7.6). The SimRank similarity matrix
$\boldsymbol{\Sigma}=(\sigma_{ij})$ of order $(n\times n)$ satisfies the matrix equation
$\boldsymbol{\Sigma}=c\boldsymbol{P}^{\text{T}}\boldsymbol{\Sigma}\boldsymbol{P}+\boldsymbol{D}$, where
$c\in(0,1)$ is the so-called delay factor, while $\boldsymbol{D}\boldsymbol{=}$diag$(\boldsymbol{d})$
with $\boldsymbol{d}=(d_1,\ldots,d_n)^{\text{T}}$ is such that $d_i\in[1-c,1]$ for
$i=1,\ldots,n$ (see e.g. \cite{liao2019second}). In addition, $\boldsymbol{P}=(p_{ij})$ is the
column-normalized adjacency matrix $\boldsymbol{A}$, i.e.
$p_{ij}=a_{ij}/\sum_{l=1}^na_{il}$, and hence $\boldsymbol{P}\geq\boldsymbol{0}$. In turn,
the following proposition provides the target result for
$\boldsymbol{\Sigma}$.

\begin{proposition}\label{prop:4}The matrix $\boldsymbol{\Sigma}$ is completely positive.
\end{proposition}
\begin{proof}The matrix equation which defines $\boldsymbol{\Sigma}$ may be
recursively solved in order to obtain the series representation
\begin{equation*}
\boldsymbol{\Sigma}=\sum_{k=0}^\infty
c^k(\boldsymbol{P}^{\text{T}})^k\boldsymbol{DP}^k\text{ .}
\end{equation*}
Since $\boldsymbol{P}\geq\boldsymbol{0}$ by definition, it holds that
$\boldsymbol{P}^k\geq\boldsymbol{0}$ on the basis of $(ii)$ of Proposition 2 and hence
$(\boldsymbol{P}^{\text{T}})^k\boldsymbol{DP}^k=(\boldsymbol{P}^k)^{\text{T}}\boldsymbol{DP}^k\in\mathcal{C}%
\mathcal{P}_n$ on the basis of $(iii)$ of Proposition 1. Since the sum of
completely positive matrices is a completely positive matrix from $(i)$ of
Proposition 1 and $\mathcal{C}\mathcal{P}_n$ is a closed cone, it follows
that $\boldsymbol{\Sigma}\in\mathcal{C}\mathcal{P}_n$.
\end{proof}

\section{Aggregation of multiplex networks} \label{sec::barycenters}
Let us consider a multilayer network where each layer is a two-mode bipartite network. We assume the additional property that the bipartite graphs on each layer have in common one of the two sets of the node bipartition. The distinct subsets of nodes can be interpreted as distinct set of items, while the vertex subset can be seen as a set of groups the different items belong (or do not belong) to. A practical example of this structure is constituted by a network of journals in a given discipline with respect to the
editors in the journal boards and to the authors contributing to the journals. These relations can be modelled as a two-layer network of the type described above. Indeed, the first layer is a bipartite graph connecting editors (items) to journals (groups); the second layer is a bipartite graph connecting authors (items) to journals (groups). Therefore, the set of journals (groups) is common to both layers. 

This particular multilayer structure can be turned into a multiplex network by considering the one-mode projection of each layer on the common subset of nodes. Figure \ref{fig::schema_multilayeronemode} displays the example of the simple two-layer journal network discussed above (Fig. \ref{fig::schema_multilayeronemode} (a)), which can be turned into a multiplex network of journals (Fig. \ref{fig::schema_multilayeronemode} (b)) through one-mode projection. 

\begin{figure}[t]
    \centering
    \includegraphics[width=\textwidth]{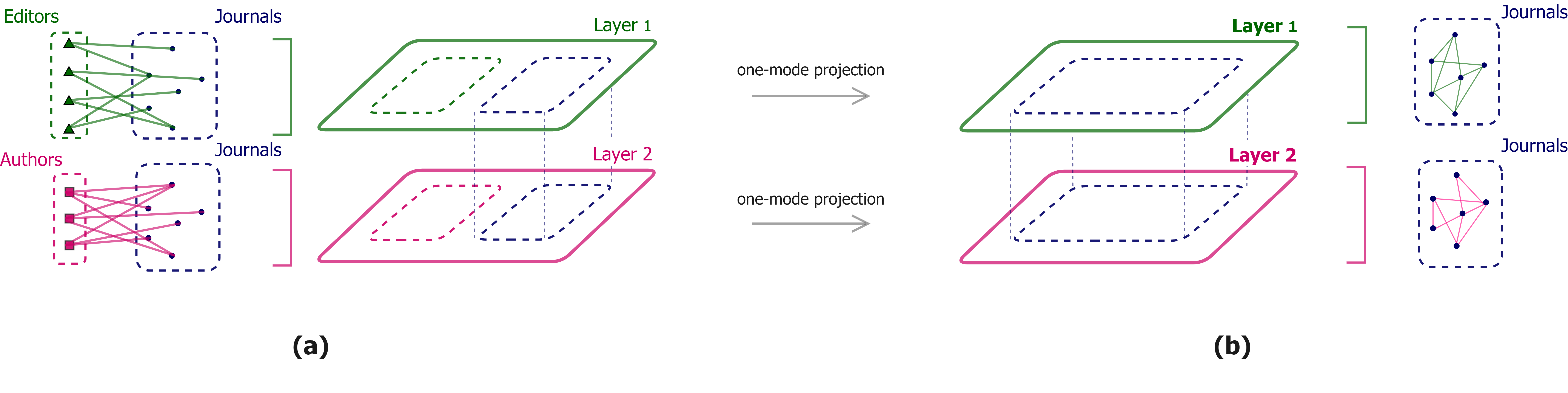}
    \caption{Scheme depicting an example of multilayer network of journals, editors and authors with two layers, each containing a bipartite graph. We start by considering a multilayer network with a bipartite network on each layer (a). In (a), the set of journals (blue nodes) is common to all the layers, while the green and pink sets represent editors and authors, respectively. The one mode projection on the set of journals allows to obtain the multiplex network (b), where nodes on each layer are journals.}
    \label{fig::schema_multilayeronemode}
\end{figure}

Let us deal with a multilayer network with $m$ layers having the structure described above. For the seek of simplicity, we will refer to the common set of nodes in all the layers as the set of groups, and to the distinct set of nodes as the set of items. 
Let $n$ denote the number of groups. Formally, $m$ two-mode bipartite networks are considered, where for the $k$-th layer
there exists an incidence matrix $\boldsymbol{B}_k=(b_{k,ij})$ of order $(p_k\times n)$, with $p_k$ denoting the set of items in the $k$-th layer,
and
\begin{equation*}
b_{k,ij}=
\begin{cases}
1&\text{if item }i\text{ of layer }k\text{ belongs to group }j\\
0&\text{otherwise ,}
\end{cases}
\end{equation*}
with $k=1,\ldots,m$. 
Thus, there are $m$ matrices $\boldsymbol{G}_k=(g_{k,ij})$ of order $(n\times n)$ given by
$\boldsymbol{G}_k=\boldsymbol{B}_k^{\text{T}}\boldsymbol{B}_k$ for $k=1,\ldots,m$. Correspondingly, there are
$k$ weighted and unweighted adjacency matrices $\boldsymbol{Z}_k=(z_{k,ij})$,
$\boldsymbol{A}_k=(a_{k,ij})$ and $\boldsymbol{F}_k=(f_{k,ij})$ of order $(n\times n)$ for
$k=1,\ldots,m$. 
Based on the matrices $\boldsymbol{G}_k$s or $\boldsymbol{F}_k$s, it is possible to define $m$ suitable similarity matrices for each layer. 
Thus, we explore the aggregation of $m$ layers of the
multiplex network into a monoplex, i.e. we assess how to achieve a suitable
synthesis of the multiplex network. This problem is pursued by averaging in some appropriate way the $m$ similarity matrices corresponding to the layers,
which will be denoted as $\boldsymbol{S}_1,\ldots, \boldsymbol{S}_m$. Specifically, three proposals for a suitable average, which will be indicated as $\boldsymbol{S}_{+}$, of the $m$ similarity
matrices $\boldsymbol{S}_1,\ldots,\boldsymbol{S}_m\in\mathcal{C}\mathcal{P}_n$ will be considered.

The computation of the barycenter of a set of objects is generally based on the minimization of an appropriate objective
function. More precisely, $\boldsymbol{S}_{+}$ is obtained as
\begin{equation*}\label{eq:baric}
\boldsymbol{S}_{+}=\arg\min_{\boldsymbol{X}\in\mathcal{P}_n}\sum_{k=1}^m w_k d(\boldsymbol{X},%
\boldsymbol{S}_k)^2\text{ ,}
\end{equation*}
where $d:\mathcal{P}_n\times\mathcal{P}_n\rightarrow\mathbb{R}^{+}$ is a
metric on $\mathcal{P}_n$, while the known weights $w_1,\ldots,w_m$ are such
that $w_k\geq 0$ for $k=1,\ldots,m$ and $\sum_{k=1}^mw_k=1$. In a general
framework, $\boldsymbol{S}_{+}$ provides the so-called Fr\a'echet mean (see e.g. \cite{bacak2014computing}). However, the Fr\a'echet mean might be not suitable in an arbitrary metric space,
although it is highly appropriate in a geodesic metric space of non-positive
curvature -- a Hadamard space -- such as $\mathcal{P}_n$. The existence
and uniqueness of the minimizer $\boldsymbol{S}_{+}$ in a Hadamard space is assured by
Theorem 2.4 by \cite{bacak2014computing}. In the following subsections, we discuss the adoption of three appropriate metrics $d$, i.e. the Frobenius, the Riemannian and the
Wasserstein metrics, and the properties of the corresponding $\boldsymbol{S}_{+}$.

\subsection{Frobenius metric}\label{subsec::frob}
If the classical Frobenius metric is adopted, i.e.
\begin{equation*}\label{def:frobmetric}
d_{\text{F}}(\boldsymbol{X},\boldsymbol{Y})=\|\boldsymbol{X}\boldsymbol{-}\boldsymbol{Y}\|_{%
\text{F}}=\text{tr}((\boldsymbol{X}-\boldsymbol{Y})^{\text{T}}(\boldsymbol{X}-\boldsymbol{Y}))^{1/2}
\end{equation*}
for $\boldsymbol{X},\boldsymbol{Y}\in\mathcal{P}_n$, the weighted mean suggested in \cite{abdi2005distatis} is obtained
\begin{equation}\label{eq::frobeniusbar}
\boldsymbol{S}_{+\text{F}}=\sum_{k=1}^m w_k\boldsymbol{S}_k\text{ .}
\end{equation}
From $(i)$ of
Proposition 1, it holds that $\boldsymbol{S}_{+\text{F}}\in\mathcal{C}\mathcal{P}_n$, since  $\boldsymbol{S}_{+\text{F}}$ is a weighted sum of completely positive matrices
with positive weights. However, this proposal may involve a ``swelling
effect'' as a drawback, in the sense that $\boldsymbol{S}_{+\text{F}}$ may show an
increase in the determinant with respect to the single components of the mean \citep{fletcher2007riemannian}.

\subsection{Riemannian metric} \label{subsec::riemann}
Let us assume the Riemannian metric, i.e.
\begin{equation*}\label{def:riemandist}
\begin{aligned}
d_{\text{R}}(\boldsymbol{X},\boldsymbol{Y}){}
&=\|\log(\boldsymbol{X}^{-1/2}\boldsymbol{YX}^{-1/2})\|_{%
\text{F}}\\
&=\text{tr}(\log(\boldsymbol{X}^{-1/2}\boldsymbol{YX}^{-1/2})^{\text{T}}\log(%
\boldsymbol{X}^{-1/2}\boldsymbol{YX}^{-1/2}))^{1/2}
\end{aligned}
\end{equation*}
for $\boldsymbol{X},\boldsymbol{Y}\in\mathcal{P}_n$. \cite{bhatia2009positive} (Chapter 6, Theorem 6.1.6)
remarks that $d_{\text{R}}$ naturally arises in the framework of Riemannian
geometry (for further details, see also \cite{bhatia2019procrustes} and
references therein). Moreover, $d_{\text{R}}$ may be considered as the matrix
version of the Fisher-Rao metric for probability laws (see in turn Chapter 6
in \cite{bhatia2009positive}). In this case, the matrix $\boldsymbol{S}_{+\text{R}}$ exists and is the
unique solution of the nonlinear matrix equation in $\boldsymbol{X}$ given by 
\begin{equation}\label{eq:nonlinear}
\sum_{k=1}^mw_k\log(\boldsymbol{X}^{1/2}\boldsymbol{S}_k^{-1}\boldsymbol{X}^{1/2})=%
\boldsymbol{0}\text{ ,}
\end{equation}
even if it is not generally expressible in closed form (see e.g. \cite{lim2014weighteda,lim2014weighted}). 

In order to introduce the explicit expression of $\boldsymbol{S}_{+\text{R}}$ for $m=2$,
let us consider the weighted geometric mean of two matrices
$\boldsymbol{X},\boldsymbol{Y}\in\mathcal{P}_n$, i.e.
\begin{equation}\label{eq:rieman2}
\boldsymbol{X}\#_w\boldsymbol{Y}=\boldsymbol{X}(\boldsymbol{X}^{-1}\boldsymbol{Y})^w=\boldsymbol{X}^{1/2}(%
\boldsymbol{X}^{-1/2}\boldsymbol{YX}^{-1/2})^w\boldsymbol{X}^{1/2}\text{ ,}
\end{equation}
where $w\in[0,1]$ (for its properties, see \cite{bhatia2009positive}, and for its
efficient computation see \cite{iannazzo2016geometric}). Obviously, if $\boldsymbol{X}$ and $\boldsymbol{Y}$ were
scalars and $w = 1/2$, expression \eqref{eq:rieman2} reduces to the usual geometric
mean of two numbers. Moreover, if $\boldsymbol{X},\boldsymbol{Y}\in\mathcal{P}_n$ it follows that
$\boldsymbol{X}\#_w\boldsymbol{Y}\in\mathcal{P}_n$ \cite{lim2014weighteda}. Actually, for $m=2$ it
can be proven that the solution of the nonlinear matrix equation \eqref{eq:nonlinear} is given by
\begin{equation}\label{eq::closedriem}
\boldsymbol{S}_{+\text{R}}=\boldsymbol{S}_2\text{\#}_{w_1}\boldsymbol{S}_1\boldsymbol{=}%
\boldsymbol{S}_2^{1/2}(\boldsymbol{S}_2^{-1/2}\boldsymbol{S}_1\boldsymbol{S}_2^{-1/2})^{w_1}%
\boldsymbol{S}_2^{1/2}
\end{equation}
(see e.g. \cite{bhatia2009positive}, p.210). This is the reason for which
$\boldsymbol{S}_{+\text{R}}$ is named as the weighted geometric mean of positive
definite matrices also for a general $m$. It should be remarked that
$\boldsymbol{S}_{+\text{R}}= \boldsymbol{S}_1$ if $w_1=1$ and $\boldsymbol{S}_{+\text{R}}=\boldsymbol{S}_2$ if $w_1=0$.
Moreover, since $\boldsymbol{S}_{+\text{R}}$ constitutes a geodesic from $\boldsymbol{S}_1$ to $\boldsymbol{S}_2$
for varying $w_1\in[0,1]$ (see e.g. \cite{bhatia2009positive}, p.210) and
$\mathcal{C}\mathcal{P}_n$ is a closed convex cone in $\mathcal{P}_n$, then
$\boldsymbol{S}_{+\text{R}}\in\mathcal{C}\mathcal{P}_n$ for $m=2$.

If $m\geq 3$, \cite{lim2014weighteda} propose an iterative procedure for
computing $\boldsymbol{S}_{+\text{R}}$. Let us define the iterative sequence 
\begin{equation}\label{eq::induction_Riem}
\boldsymbol{X}_{k+1}=\boldsymbol{X}_k \#_{t_{k+1}}\boldsymbol{S}_{j_{k+1}}
\end{equation}
for $k=1,2,\ldots,$ where $\boldsymbol{X}_1 = \boldsymbol{S}_{j_1}$ and $t_k=w_{j_k}/\sum_{i=1}^kw_{j_i}$.
In addition, $j_k=(k\bmod m)$ and null residuals are identified with $m$, i.e. $j_{im}=m$ for
$i=1,2,\ldots$. Accordingly, \cite{lim2014weighteda} prove that $\lim_k\boldsymbol{X}_k=\boldsymbol{S}_{+\text{R}}$. In \cite{massart2018matrix}, an alternative choice leading to a more efficient computation of the index $j_k$ is provided. Since at each step
$\boldsymbol{X}_k\in\mathcal{C}\mathcal{P}_n$, it also follows that
$\boldsymbol{S}_{+\text{R}}\in\mathcal{C}\mathcal{P}_n$. Moreover,
$\boldsymbol{S}_{+\text{R}}\preceq \boldsymbol{S}_{+\text{F}}$ holds on the basis of the generalized
arithmetic-geometric-harmonic mean inequality (see \cite{lim2014weighted}),
a result which emphasizes that $\boldsymbol{S}_{+\text{R}}$ may be less prone than
$\boldsymbol{S}_{+\text{F}}$ to the swelling effect.

\subsection{Wasserstein metric}\label{subsec::wass}
A further proposal for $d$ is given by the
Wasserstein metric, i.e.
\begin{equation*}\label{def:wdist}
d_{\text{W}}(\boldsymbol{X},\boldsymbol{Y})=\text{tr}(\boldsymbol{X}+\boldsymbol{Y}-2(\boldsymbol{X}^{1/2}%
\boldsymbol{YX}^{1/2})^{1/2})^{1/2}
\end{equation*}
for $\boldsymbol{X},\boldsymbol{Y}\in\mathcal{P}_n$ (for more details, see e.g. \cite{bhatia2019bures}).
\cite{bhatia2019bures} emphasize that $d_{\text{W}}$ displays many interesting
features and, among others, it corresponds to a metric in Riemannian
geometry. \cite{alvarez2016fixed} prove that $\boldsymbol{S}_{+\text{W}}$ is
provided by the unique solution of the nonlinear matrix equation in $\boldsymbol{X}$ given by 
\begin{equation}\label{eq:nonlinear2}
\boldsymbol{X}=\sum_{k=1}^mw_k(\boldsymbol{X}^{1/2}\boldsymbol{S}_k\boldsymbol{X}^{1/2})^{1/2}
\end{equation}
and $\boldsymbol{S}_{+\text{W}}\in\mathcal{P}_n$ (see also \cite{bhatia2019bures}). The
solution of the nonlinear matrix equation \eqref{eq:nonlinear2} is solely known in a
closed form for $m=2$ and it reads
\begin{equation}\label{eq::closedwass}
\boldsymbol{S}_{+\text{W}}=w_1^2\boldsymbol{S}_1+w_2^2\boldsymbol{S}_2+w_1w_2((\boldsymbol{S}_1%
\boldsymbol{S}_2)^{1/2}+(\boldsymbol{S}_2\boldsymbol{S}_1)^{1/2})
\end{equation}
(see \cite{bhatia2019bures}). Incidentally, it is interesting to remark that, if
$\boldsymbol{S}_1$ and $\boldsymbol{S}_2$ were scalars and $w_1 = w_2 = 1/2$, expression \eqref{eq::closedwass}
actually provides the average of the usual arithmetic and geometric means of
two numbers. Similarly to the case of the weighted geometric mean, it holds
that $\boldsymbol{S}_{+\text{W}}=\boldsymbol{S}_1$ if $w_1=1$ and $\boldsymbol{S}_{+\text{R}}=\boldsymbol{S}_2$ if $w_1=0$,
while $\boldsymbol{S}_{+\text{W}}$ is a geodesic from $\boldsymbol{S}_1$ to $\boldsymbol{S}_2$ for varying
$w_1\in[0,1]$ \citep{bhatia2019bures} and hence
$\boldsymbol{S}_{+\text{R}}\in\mathcal{C}\mathcal{P}_n$ for $m=2$. 

In order to manage the case $m\geq 3$ by adopting an algorithm based on a
fixed-point iteration method, \cite{alvarez2016fixed} suggest to
consider the matrix function $\boldsymbol{K}:\mathcal{P}_n\rightarrow\mathcal{P}_n$ such
that
\begin{equation}\label{eq::update_Wass}
\boldsymbol{K}(\boldsymbol{X})=\boldsymbol{X}^{-1/2}\left(\sum_{i=1}^m w_i(\boldsymbol{X}^{1/2}%
\boldsymbol{S}_i\boldsymbol{X}^{1/2})^{1/2}\right)^2\boldsymbol{X}^{-1/2}\text{ .}
\end{equation}
\cite{alvarez2016fixed} highlight that $\boldsymbol{K}(\boldsymbol{X})\in\mathcal{P}_n$ if
$\boldsymbol{X}\in\mathcal{P}_n$. Hence, by assuming that $\boldsymbol{X}_{k+1}=\boldsymbol{K}(\boldsymbol{X}_k)$ for
$k=0,1,\ldots$ and for each $\boldsymbol{X}_0\in\mathcal{P}_n$, \cite{alvarez2016fixed} prove that $\lim_k\boldsymbol{X}_k=\boldsymbol{S}_{+\text{W}}$. In addition, by means of a
numerical study, they remark that algorithm
convergence is fast, even for rather large $n$ and $m$. \ It is worth
noticing that $\boldsymbol{S}_{+\text{W}}\preceq\boldsymbol{S}_{+\text{F}}$ on the basis of Theorem
9 by \cite{bhatia2009positive}, which is a suitable property, as previously
explained.
\\\\
\textbf{Remark.} The use of the metrics $d_{\text{R}}$
and $d_{\text{W}}$ involves an optimal rotation of each couple of similarity
matrices, i.e. they actually provide the minimum for the orthogonal
Procrustes problem (see comment to Theorem 1 in \cite{bhatia2009positive} and
\cite{bhatia2019procrustes}). Moreover, we also remark that 
$\boldsymbol{S}_{+\text{R}}$ and $\boldsymbol{S}_{+\text{W}}$ have to be appropriately re-normalized in
order to achieve a proper normalized similarity matrix having all values equal to $1$ on the diagonal. A suitable way to achieve a normalized similarity matrix $\boldsymbol{S^*}=(s_{ij}^*)$ consists in modifying the generic element $s_{ij}$ of a similarity matrix $\boldsymbol{S}=(s_{ij})$ as follows
\begin{equation*}
    \label{eq:normalizationsimmat}
    s_{ij}^* = \frac{s_{ij}}{(s_{ii}s_{jj})^{1/2}} \text{ .}
\end{equation*} 
\\
\section{Choice of the weights}\label{sec::weightchoice} 
In order to select the weights
$\boldsymbol{w}=(w_1,\ldots,w_m)^{\text{T}}$, a measure of `closeness' between
the couples of the $m$ similarity matrices $\boldsymbol{S}_1,\ldots,\boldsymbol{S}_m$ is needed at
first. A suitable such a measure is given by the RV coefficient
proposed by \cite{robert1976unifying}, which is widely adopted in many
different practical analyses. Hence, the matrix $\boldsymbol{R} =(r_{ij})$ of order
$(m\times m)$ may be considered, where $r_{ij}$ represents the RV coefficient
between $\boldsymbol{S}_i$ and $\boldsymbol{S}_j$, i.e.
\begin{equation*}
r_{ij}=\frac{\langle\boldsymbol{S}_i,\boldsymbol{S}_j\rangle_{\text{F}}}{\|\boldsymbol{S}_i\|_{%
\text{F}}\|\boldsymbol{S}_j\|_{\text{F}}}\text{ .}
\end{equation*}
It holds that $r_{ij}\in[0,1]$ and the closeness between $\boldsymbol{S}_i$ and $\boldsymbol{S}_j$
increases as $r_{ij}$ approaches one. In addition, it should be remarked that
$\boldsymbol{R}\in\mathcal{C}\mathcal{P}_m$. Indeed, this issue follows from $(iii)$ of
Proposition 1, since $\boldsymbol{R}$ may be expressed as
$\boldsymbol{R} = \text{diag}(\boldsymbol{U})^{-1/2} \boldsymbol{U} \text{diag}(\boldsymbol{U})^{-1/2}$ where $\boldsymbol{U}=\boldsymbol{V}^{\text{T}}\boldsymbol{V}$ and
$\boldsymbol{V}=($vec$(\boldsymbol{S}_1),\ldots,$vec$(\boldsymbol{S}_m))$, which also provides a computationally
efficient expression for $\boldsymbol{R}$.

When $\boldsymbol{S}_{+\text{F}}$ is adopted, \cite{abdi2005distatis} suggest to consider the
eigendecomposition of $\boldsymbol{R}$. In such a case, the decomposition gives rise to $\boldsymbol{R}=\boldsymbol{Q}\boldsymbol{\Lambda}\boldsymbol{Q}^{\text{T}}$, where
$\boldsymbol{Q}=(\boldsymbol{q}_1,\ldots,\boldsymbol{q}_m)$ is the orthogonal matrix whose columns are the
eigenvectors of $\boldsymbol{R}$, while
$\boldsymbol{\Lambda}=$diag$(\lambda_1(\boldsymbol{R}),\ldots,\lambda_m(\boldsymbol{R}))$ is the
diagonal matrix whose entries are the positive eigenvalues of $\boldsymbol{R}$ considered in
nonincreasing order. The Perron-Frobenius theorem (see e.g. \cite{berman1994nonnegative}) ensures that the eigenvector $\boldsymbol{q}_1$ corresponding to the
largest eigenvalue $\lambda_1(\boldsymbol{R})$ has nonnegative elements. Thus, \cite{abdi2005distatis} propose the choice 
\begin{equation*}
\boldsymbol{w}=\frac{1}{\boldsymbol{1}^{\text{T}}\boldsymbol{q}_1}\,\boldsymbol{q}_1\text{ .}
\end{equation*}
In practice, a principal component analysis is considered on $\boldsymbol{R}$ and the first
eigenvector is used for implementing $\boldsymbol{w}$. Hence, since layers with larger
projections on $\boldsymbol{q}_1$ are more similar to the other layers than the layers
with smaller projections, the (rescaled) elements of this eigenvector should
provide suitable weights for $\boldsymbol{S}_{+\text{F}}$, which is a linear combination
of $\boldsymbol{S}_1,\ldots,\boldsymbol{S}_m$. Similarly to principal component analysis, the
goodness of this selection for $\boldsymbol{w}$ may be assessed by means of the ratio
$\lambda_1(\boldsymbol{R})/\sum_{k=1}^m\lambda_k(\boldsymbol{R})$ (see \cite{abdi2005distatis}).

In the case of $\boldsymbol{S}_{+\text{R}}$ and $\boldsymbol{S}_{+\text{W}}$ it is not obvious if the
previous choice of $\boldsymbol{w}$ is suitable, since these averages are not linear
functions of $\boldsymbol{S}_1,\ldots,\boldsymbol{S}_m$. Alternatively, we suggest the choice
\begin{equation*}
\boldsymbol{w}=\frac{1}{\boldsymbol{1}^{\text{T}}(\boldsymbol{R}-\boldsymbol{I}_m)%
\boldsymbol{1}}\,(\boldsymbol{R}-\boldsymbol{I}_m)\boldsymbol{1} \text{ .}
\end{equation*}
The rationale behind this proposal stems on the fact that
the more a similarity matrix is close to the others, the more is
representative of the whole set of similarity matrices -- and hence it should
receive a larger weight with respect to the others. Thus, this proposal
assigns weights according to the ratio of the RV coefficient sum corresponding
to a similarity matrix to the total of the RV coefficients.

\section{Implementation of algorithms for computing averages}\label{sec::implementation}

The algorithms for the computation of the similarity matrix averages $\boldsymbol{S}_{+\text{F}}$, $\boldsymbol{S}_{+\text{W}}$ and $\boldsymbol{S}_{+\text{R}}$ were implemented using the Python programming language. All the Python functions are publicly available on GitHub at \url{https://github.com/DedeBac/SimilarityMatrixAggregation.git}.

The choice of the Frobenius metric leads to a similarity matrix average which is actually a weighted arithmetic mean of matrices -- see expression \eqref{eq::frobeniusbar}. Hence, the practical implementation of $\boldsymbol{S}_{+\text{F}}$ is trivial. In contrast, as discussed in Subsections \ref{subsec::riemann} and \ref{subsec::wass}, the choice of the Riemannian and the Wasserstein metrics requires that $\boldsymbol{S}_{+\text{R}}$ and $\boldsymbol{S}_{+\text{W}}$ are computed via iterative algorithms, unless there are solely two input matrices. Thus, in the special case $m=2$ (discussed in Subsection \ref{subsec::riemann} and \ref{subsec::wass}), the closed-form solutions available in expressions \eqref{eq::closedriem} and \eqref{eq::closedwass} were implemented. In the non-trivial case $m>2$, $\boldsymbol{S}_{+\text{R}}$ is computed using the iterative sequence defined in expression \eqref{eq::induction_Riem}. In turn, such an implementation is based on the optimized procedure proposed by \cite{massart2018matrix}, as introduced in Subsection \ref{subsec::riemann}. A key issue of the iterative step consists in the computation of the weighted geometric mean of two matrices. In our algorithm, the weighted geometric mean was implemented by using the proposal by \cite{iannazzo2016geometric} (Algorithm $3.1$). Specifically, this implementation adopts the Schur decomposition and the Cholesky factorization in order to simplify the computation of the power of a matrix. In the case of the the Wasserstein metric, $\boldsymbol{S}_{+\text{W}}$ is computed by means of the fixed-point iteration method proposed by \cite{alvarez2016fixed} and discussed in Subsection \ref{subsec::wass}, see equation \eqref{eq::update_Wass}. Both these iterative procedures stop when either a maximum number of iterations is reached, or when the relative distance between intermediate solutions computed at consecutive iterations reaches the tolerance value set by the user. 
Finally, the implementation of the weights computation follows the proposals discussed in Section \ref{sec::weightchoice}.

We provide a visual example of the computation of the matrix averages for each metric choice in the very special case $n=2$ and $m=2$. In such a setting, the results given in the previous sections may be conveniently
-- and instructively -- visualized by means of ellipses of type
$\{\boldsymbol{x}\in\mathbb{R}^2:\boldsymbol{x}^{\text{T}}\boldsymbol{Sx}=1\}$. Thus, for $w_1=w_2=1/2$, Figures \ref{fig:toyex1}, \ref{fig:toyex2} and \ref{fig:toyex3} provide the graphical representation of
$\boldsymbol{S}_{+\text{F}}$, $\boldsymbol{S}_{+\text{R}}$ and $\boldsymbol{S}_{+\text{W}}$ for several choices of
$\boldsymbol{S}_1$ and $\boldsymbol{S}_2$. From these Figures, it is apparent that $\boldsymbol{S}_{+\text{R}}$ and
$\boldsymbol{S}_{+\text{W}}$ are less affected by the swelling effect (see especially
Figure \ref{fig:toyex3}). Moreover, even if $\boldsymbol{S}_{+\text{F}}$, $\boldsymbol{S}_{+\text{R}}$ and
$\boldsymbol{S}_{+\text{W}}$ have rather similar ``shapes'' in Figures \ref{fig:toyex1} and \ref{fig:toyex2}, these
``shapes'' substantially differ in Figure \ref{fig:toyex3}.

\begin{figure}[!ht]
    \centering
    \includegraphics[width=0.95\textwidth]{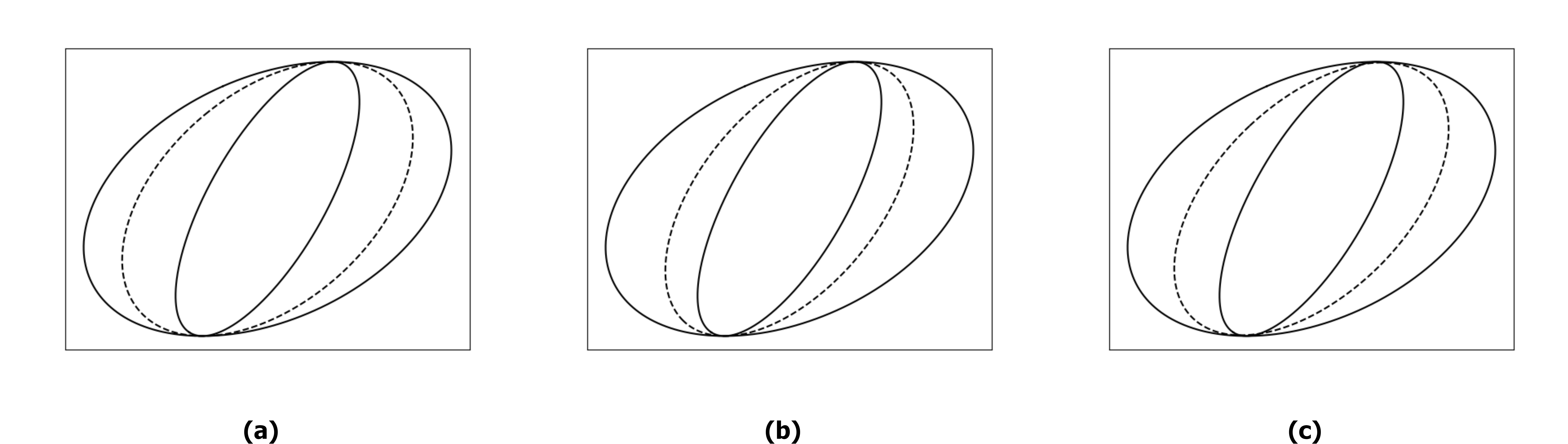}
    \caption{Visualization of (a) $\boldsymbol{S}_{+\text{F}}$, (b) $\boldsymbol{S}_{+\text{R}}$  and (c) $\boldsymbol{S}_{+\text{W}}$  (dashed style) for $\boldsymbol{S}_1=
\begin{pmatrix}
1&1\\
1&2
\end{pmatrix}
$ and $\boldsymbol{S}_2=
\begin{pmatrix}
4&1\\
1&2
\end{pmatrix}
$. Continuous lines correspond to the ellipses individuated by $\boldsymbol{S}_1$ and $\boldsymbol{S}_2$.}
    \label{fig:toyex1}
\end{figure}

\begin{figure}[!ht]
    \centering
    \includegraphics[width=0.95\textwidth]{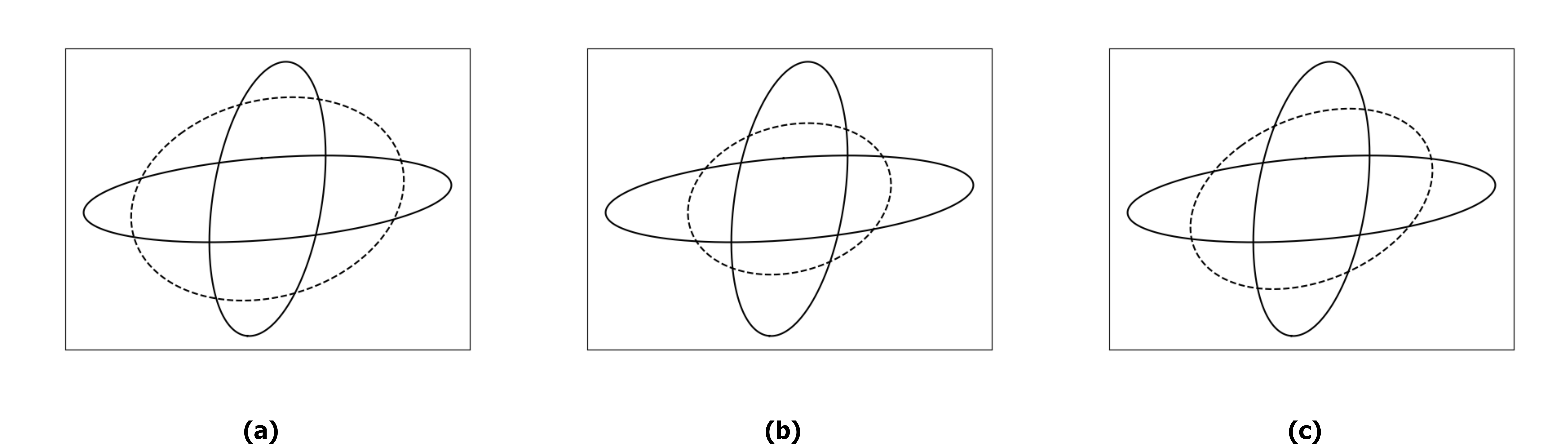}
    \caption{Visualization of (a) $\boldsymbol{S}_{+\text{F}}$, (b) $\boldsymbol{S}_{+\text{R}}$ and (c) $\boldsymbol{S}_{+\text{W}}$ (dashed style) for $\boldsymbol{S}_1=
\begin{pmatrix}
1&1\\
1&10
\end{pmatrix}
$ and $\boldsymbol{S}_2=
\begin{pmatrix}
10&1\\
1&1
\end{pmatrix}
$. Continuous lines correspond to the ellipses individuated by $\boldsymbol{S}_1$ and $\boldsymbol{S}_2$.}
    \label{fig:toyex2}
\end{figure}

\begin{figure}[!ht]
    \centering
    \includegraphics[width=0.95\textwidth]{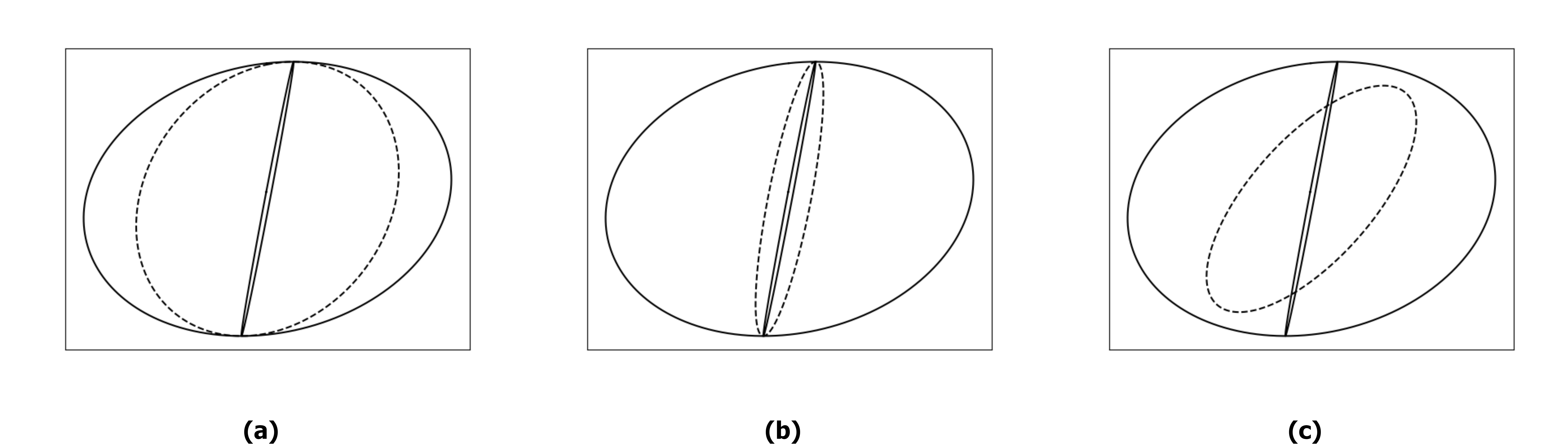}
    \caption{Visualization of (a) $\boldsymbol{S}_{+\text{F}}$, (b) $\boldsymbol{S}_{+\text{R}}$ and (c) $\boldsymbol{S}_{+\text{W}}$ (dashed style) for $\boldsymbol{S}_1=
\begin{pmatrix}
50&1\\
1&1
\end{pmatrix}
$ and $\boldsymbol{S}_2=
\begin{pmatrix}
1.01&1\\
1&1
\end{pmatrix}
$. Continuous lines correspond to the ellipses individuated by $\boldsymbol{S}_1$ and $\boldsymbol{S}_2$.}
    \label{fig:toyex3}
\end{figure}

\section{Case study: multiplex network of statistical journals}\label{sec::journals}

In this Section we present an application of the use of matrix averages to multiplex network of statistical journals. 
Specifically, we will consider 79 journals belonging to the category ``Statistics and Probability" of the \textit{Journal Citation Reports} in the year 2006. Each journal was associated with three different types of ``items" in order to obtain the multiplex network: the editors sitting on the editorial boards, the authors that published on the journals and the journals cited in the bibliographies of the corresponding articles. This kind of connections have a key relevance in a scientometric setting (see \cite{baccini2020intellectual}, \cite{baccini2022similarity} and references therein).

Editor data include the editors of the journals in the year 2006 and were collected by \cite{baccini_eds_data_2009}. On the other hand, data on authors and cited journals were extracted from Clarivate Analytics Web of Science (WoS) database (\url{https://www.webofscience.com/}). Queries were made through the SQL relational database system hosted by the Centre for Science and Technology Studies (CWTS) at Leiden University, using the 2021 version of Web of Science. In both cases, authorships and cited journals were solely collected for publications in the five-year period 2006-2010. Table \ref{table::desc_stats} offers an overview of the considered dataset.

\begin{table}[!t]
\caption{Main features of the considered scientometric dataset.}
\
\centering
\scriptsize

{}
\begin{adjustbox}{max width=\textwidth}
\begin{tabular}{cccc}
\hline 
Layer & Period of observation & Two-mode network & Relation in the one-mode network \\
\hline
Editors & 2006 & 79 journals $\times$ 2,227 editors & Interlocking editorship\\
Authors & 2006-2010 & 79 journals $\times$ 38,683 authors & Interlocking authorship\\
Cited journals & 2006-2010 & 79 journals $\times$ 7,528 cited journals & Journal bibliographic coupling \\
\hline
\end{tabular} 
\end{adjustbox}
\label{table::desc_stats}
\end{table}

Before proceeding to the computation of averages, the journal \textit{Communications in Statistics - Simulation and Computation} was removed from the dataset since it had exactly the same editors of the journal \textit{Communications in Statistics - Theory and Methods}. In fact, Taylor \& Francis, the publisher of the two journals, considers them as ``associated journals". Thus, in order to avoid ambiguity we solely retained \textit{Communications in Statistics - Theory and Methods}.

The resulting one-mode networks respectively model the pairwise connections between journals based on common editors (interlocking editorship network), common authors (interlocking authorship network) and common cited journals (journal bibliographic coupling network). To the aim of computing the analyzed averages, the cosine similarity defined in equation \eqref{eq::cossim} was considered for each layer in order to obtain three similarity matrices. 
The averages $\boldsymbol{S}_{+\text{F}}$, $\boldsymbol{S}_{+\text{R}}$ and $\boldsymbol{S}_{+\text{W}}$ of the three similarity matrices were computed using the implementation described in Section \ref{sec::implementation}.
Figure \ref{fig::residuals} displays the residual between matrices computed at consecutive iterations in the case of $\boldsymbol{S}_{+\text{R}}$ (Fig. \ref{fig::residuals} (a)) and $\boldsymbol{S}_{+\text{W}}$ (Fig. \ref{fig::residuals} (b)). It should be remarked that both algorithms converge to a solution after few iterative steps.

Figures \ref{fig::bar_frob}-\ref{fig::bar_wass} respectively display the networks corresponding to the averages $\boldsymbol{S}_{+\text{F}}$, $\boldsymbol{S}_{+\text{R}}$ and $\boldsymbol{S}_{+\text{W}}$. The networks were plotted by means of the Gephi software \citep{gephi} and by using the ForceAtlas2 visualization algorithm \citep{jacomy2014forceatlas2}.
Different colours correspond to distinct communities computed with the Louvain algorithm \citep{blondel2008fast} by setting the resolution parameter to $1$ \citep{lambiotte2008laplacian}. Actually, the Louvain algorithm is commonly adopted in order to detect communities in a network. The modularity score obtained was $0.131$ for $\boldsymbol{S}_{+\text{F}}$, $0.222$ for $\boldsymbol{S}_{+\text{R}}$, and $0.166$ for $\boldsymbol{S}_{+\text{W}}$, as reported in Table \ref{table::communities}.
The Louvain algorithm partitions the three aggregated networks into five communities. All the networks are characterized by three large communities, indicated by the pink, orange and green nodes in Figures \ref{fig::bar_frob}--\ref{fig::bar_wass}, and reported in Table \ref{table::communities} as Community 1, Community 2 and Community 3, respectively. Specifically, Community 1 (pink nodes) gathers journals in the field of statistical methodology, Community 2 (orange nodes) groups journals mainly devoted to probability theory and its applications, while Community 3 gathers journals from the field of applied statistics.  
Finally, two much smaller communities are individuated - Community 4 (blue nodes) and Community 5 (yellow nodes) in Table \ref{table::communities}. They both gather journals that are at the boundary with the fields of Economics (Community 4) and Bioinformatics (Community 5), respectively.

Overall, the modularity values and the communities detected in the three aggregated networks are very similar. This issue indicates that there are not structurally significant difference between the three averages in the present application. 
Interestingly, the resulting communities are coherent with those obtained by \cite{baccini2020intellectual} and \cite{baccini2022similarity} by adopting different methodologies. 

\begin{table}[!t]
\caption{Number of nodes for each community individuated by the Louvain modularity optimization algorithm and for each average $\boldsymbol{S}_{+\text{F}}$, $\boldsymbol{S}_{+\text{R}}$ and $\boldsymbol{S}_{+\text{W}}$. The first column reports the modularity score associated to each partitioning of the networks.}
\
\centering
\scriptsize

{}
\begin{adjustbox}{max width=\textwidth}
\begin{tabular}{ccccccc}
\hline 

Average & Modularity & Community 1 & Community 2 & Community 3 & Community 4 & Community 5 \\
\hline
$\boldsymbol{S}_{+\text{F}}$ & $0.131$ & $27$ & $22$ & $21$ & $4$ & $4$ \\ 
\
$\boldsymbol{S}_{+\text{R}}$ & $0.222$ & $29$ & $20$ & $20$ & $5$ & $4$ \\ 
\
$\boldsymbol{S}_{+\text{W}}$ & $0.166$ & $29$ & $21$ & $16$ & $8$ & $4$ \\ 
\hline
\end{tabular}
\end{adjustbox}
\label{table::communities}
\end{table}

\begin{figure}
\includegraphics[width=\textwidth]{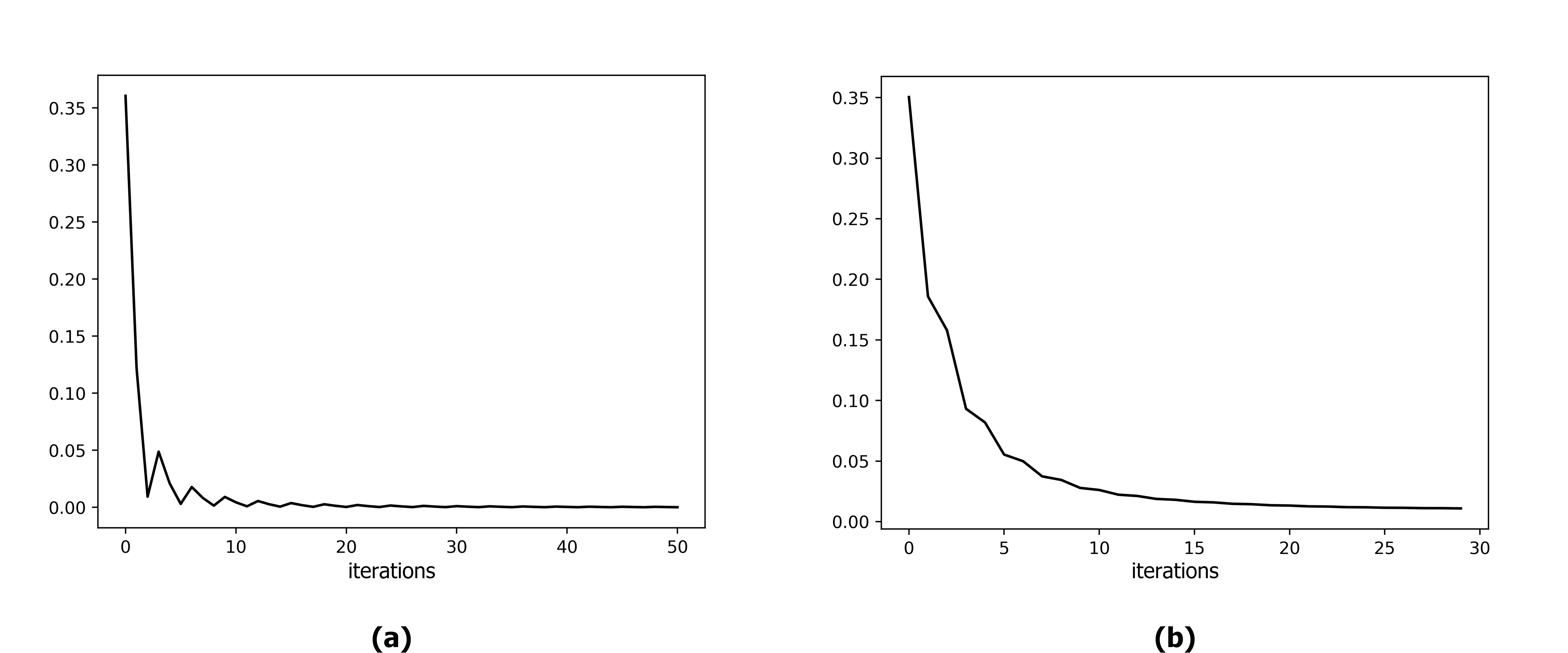}
    \caption{Residual values of the algorithm iterative steps for the computation of $\boldsymbol{S}_{+\text{W}}$ for $50$ iterations (a) and $\boldsymbol{S}_{+\text{W}}$ for $30$ iterations (b).}
    \label{fig::residuals}
\end{figure}
\newpage
\begin{sidewaysfigure}
    \centering
    \includegraphics[width=\textwidth]{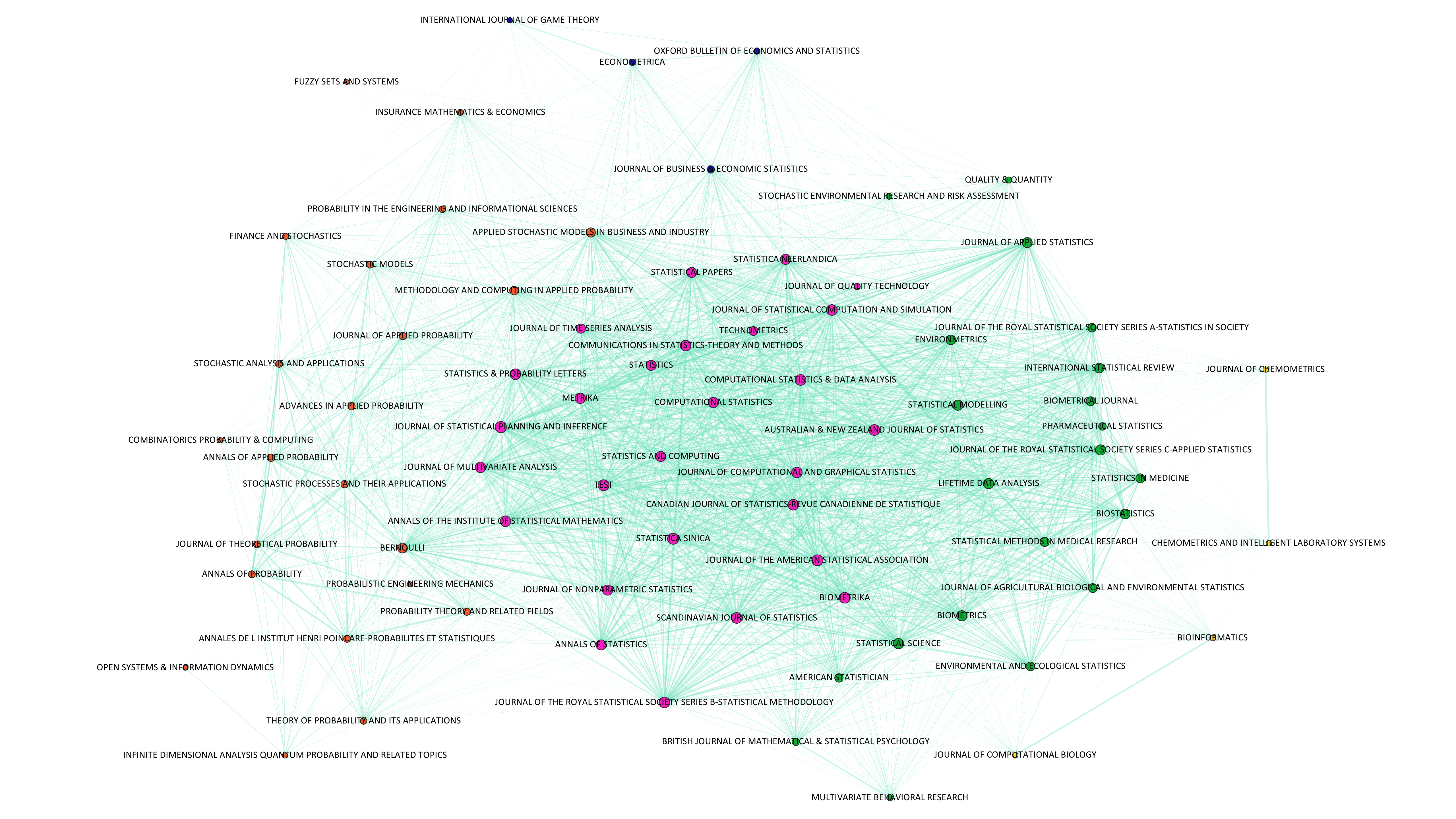}
    \caption{Visualization of the average of the journal multiplex network with the choice of the Frobenius metric. Different node colours indicate distinct communities individuated by using the Louvain optimization algorithm. The node size is proportional to its weighted degree.}
    \label{fig::bar_frob}
\end{sidewaysfigure}

\begin{sidewaysfigure}
    \centering
    \includegraphics[width=\textwidth]{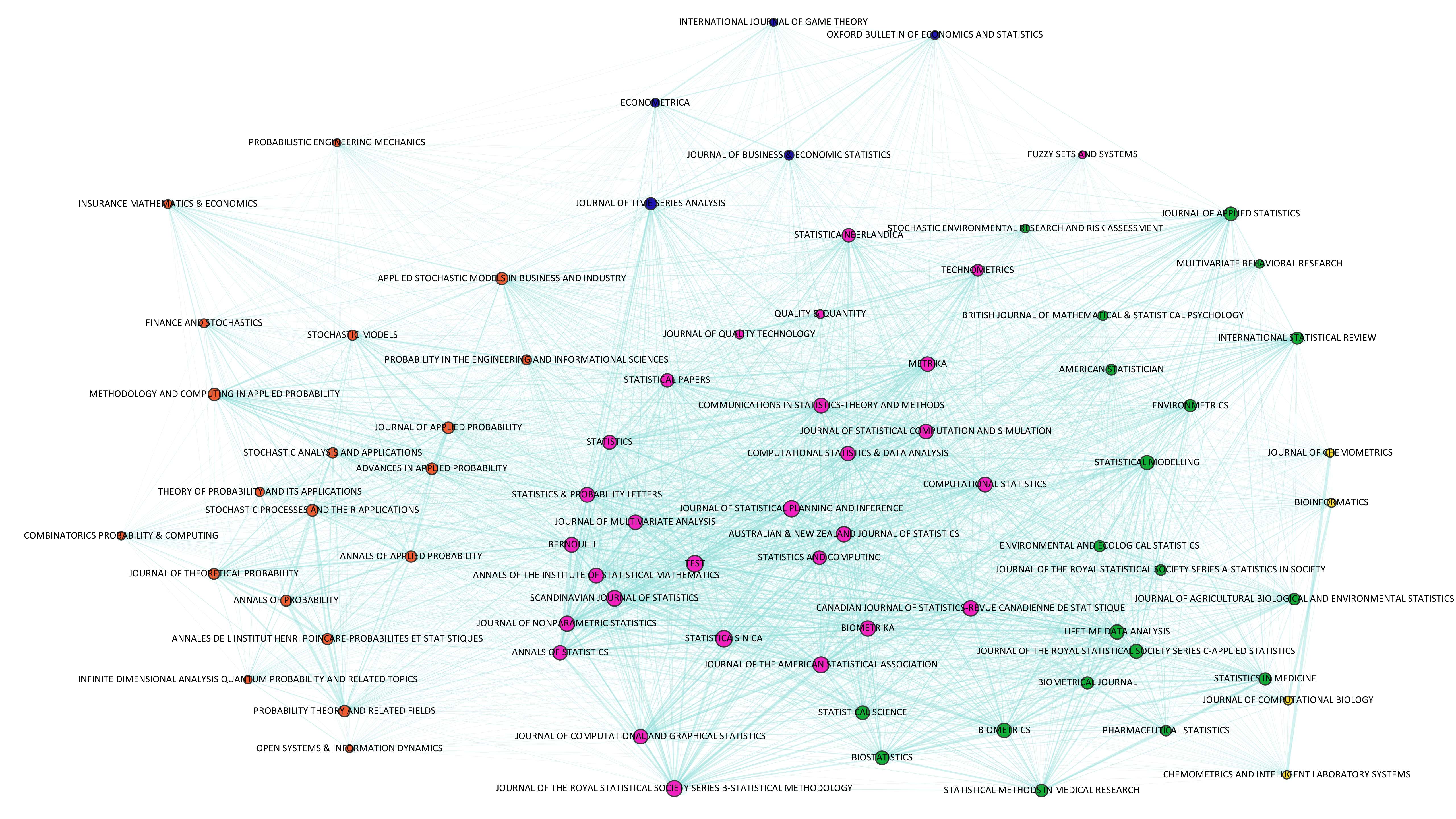}
    \caption{Visualization of the average of the jounal multiplex network with the choice of the Riemannian metric. Different node colours indicate distinct communities individuated by using the Louvain optimization algorithm. The node size is proportional to its weighted degree.}
    \label{fig::bar_riem}
\end{sidewaysfigure}

\begin{sidewaysfigure}
    \centering
    \includegraphics[width=\textwidth]{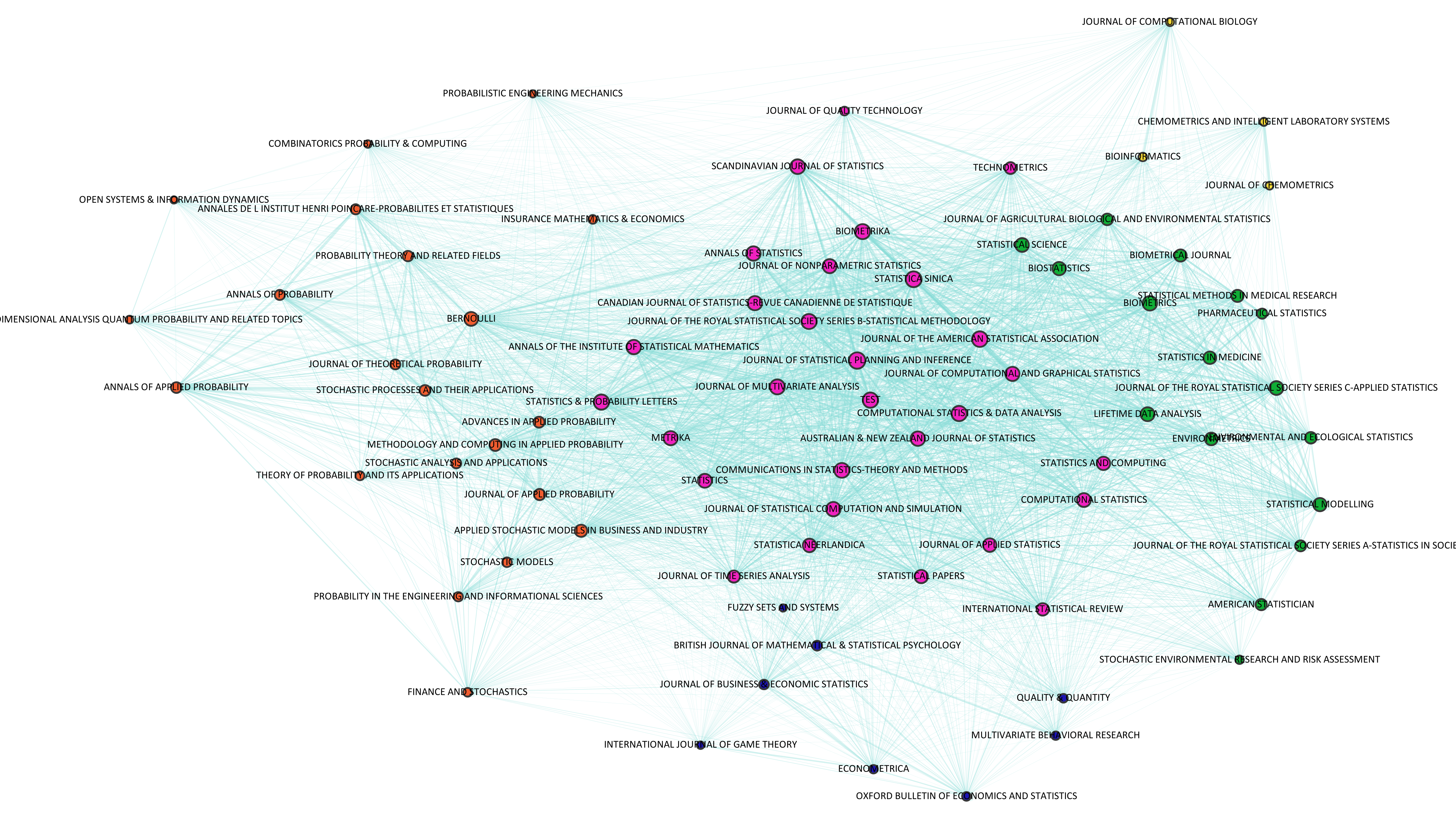}
    \caption{Visualization of the average of the journal multiplex network with the choice of the Wasserstein metric. Different node colours indicate distinct communities individuated by using the Louvain optimization algorithm. The node size is proportional to its weighted degree.}
    \label{fig::bar_wass}
\end{sidewaysfigure}

\section{Conclusions}\label{sec::conclusion}

This work proposes the use of similarity matrix average to aggregate multiplex networks by using some concepts in the Riemannian geometry formulation of barycenter. In particular, we show that some commonly-adopted similarity measures allow to obtain similarity matrices which belong to the space of completely positive matrices by starting from two-mode bipartite networks. 
These types of similarity matrices are averaged by considering the minimization of a Fr\a'echet mean criterion with the Frobenius, the Riemannian and the Wasserstein metric choices. The results obtained on the multiplex network of statistical journals highlight the availability of the methodology.  
The proposed method constitutes an advance in the topic of multiplex network aggregation, since it provides a theoretically justified framework -- along with its implementation -- to combine the contribution of different relations which exist among a set of nodes. This methodology might be of interest in several application fields, such as social network analysis and bioinformatics, where relations of different nature have to be explored and integrated to determine some structural organization of a set of entities. Indeed, the aggregation allows to perform the required task -- such as the detection of communities -- on a single network, rather than on a multiplex network, in order to reduce the complexity of the problem.

\backmatter

\newpage
\bmhead{Acknowledgments}

Eugenio Petrovich is grateful to the Centre for Science and Technology Studies (CWTS) at Leiden University for hosting him as a guest researcher and giving him access to the CWTS database system.
Lucio Barabesi gratefully acknowledges the funding from the Italian Ministry of University, PRIN project: 2017MPXW98.
\bmhead{Declarations}
The authors have no relevant financial or non-financial interests to disclose.

\bibliography{sn-bibliography}


\end{document}